\documentclass[11pt, a4paper]{article}

\usepackage[utf8]{inputenc}
\usepackage[T1]{fontenc}
\usepackage{amsmath, amssymb, amsthm, amsfonts}
\usepackage{graphicx}
\usepackage{natbib}
\usepackage[margin=2.5cm]{geometry}
\usepackage{booktabs}
\usepackage{hyperref}
\usepackage{bm} 
\usepackage{setspace}
\usepackage{xcolor}
\usepackage{multirow}
\usepackage{subcaption}
\usepackage{mathtools}  
\hypersetup{
    colorlinks=true,
    linkcolor=blue,
    citecolor=blue,
    urlcolor=blue,
    pdfauthor={Mehmet Sıddık Çadırcı},
pdftitle={A Fractional M/M/1 Queue Governed by Stretched Non-Local Time Operators}
}

\theoremstyle{plain}
\newtheorem{theorem}{Theorem}
\newtheorem{lemma}{Lemma}
\newtheorem{proposition}{Proposition}
\newtheorem{corollary}{Corollary}

\theoremstyle{definition}
\newtheorem{definition}{Definition}

\title{\textbf{A Fractional M/M/1 Queue Governed by Stretched Non-Local Time Operators}}

\author{
    Mehmet Sıddık Çadırcı$^{1,}$\thanks{Corresponding author. Email: \texttt{msiddikcadirci@cumhuriyet.edu.tr}}
}

\date{
    \slshape\small
    $^1$ Faculty of Science, Department of Statistics, Cumhuriyet University, Sivas, Türkiye.\\[3ex]
    \normalfont
    \today
}

\begin{document}

\maketitle

\begin{abstract}
We introduce a non-Markovian generalization of the classical M/M/1 queue by incorporating extended nonlocal time dynamics into Kolmogorov forward equations.
We obtain the model by replacing the standard time derivative with an extended Caputo-type operator. It preserves the birth-death transition structure of the standard queue while introducing memory effects into the temporal evolution.
We derive explicit representations for transient state probabilities in terms of the Kilbas-Saigo function, which naturally emerges as the relaxation kernel associated with the stretched operator, using Laplace transform techniques.
We construct a time-varying interpretation and show that the fractional queue can be viewed as a distribution of a classical M/M/1 process evaluated at a non-decreasing random time.
It is observed that the fractional queue can be viewed as a distribution of a classical M/M/1 process evaluated at a non-decreasing random time. We prove that under the standard stability condition $\rho<1$, the steady-state distribution remains geometric and coincides with the distribution of the classical queue, whilst we prove that the stretched fractional parameters significantly affect the convergence rate in the transient regime.
Numerical examples based on Monte Carlo simulations
highlight the effect of the parameters $(\alpha,\gamma)$ on the distribution of empty states, tail length distributions, and the average tail evolution, and validate the flexibility of the proposed framework in capturing long-memory tail dynamics.


\end{abstract}

\section{Introduction} \label{sec:intro} 

Queueing models are a canonical framework for characterizing congestion and delays in communication systems, service operations, and stochastic networks.
For the classical setting, the M/M/1 queue has a Markovian nature, and it is governed by the linear Kolmogorov forward equations, where the transition probabilities are determined by an ordinary derivative with respect to time.
Empirical dynamics in complex systems, however, often exhibit nonlocal temporal effects
and persistence, which encourages non-Markov generalizations based on fractional calculus.

Fractional time operators are now a standard tool for modeling memory and anomalous relaxation phenomena;
comprehensive reviews
can be seen in \cite{Podlubny1999,Diethelm2010,KilbasSrivastavaTrujillo2006}.
An important probabilistic route connects fractional evolution equations to time changes by
inverse subordinators and related non-decreasing processes; see, e.g.,
\cite{MeerschaertScheffler2004,MeerschaertNaneVellaisamy2011,SchillingSongVondracek2010}. These stochastic representations deliver options for fractional Cauchy problems using
inverse subordinators and time-changed semigroups
\cite{BaeumerMeerschaert2001,Toaldo2015}.

According to renewal theory, fractional relaxations naturally contribute to fractional generalizations of
Poisson-type counting frameworks \cite{MainardiGorenfloScalas2004,OrsingherBeghin2004,LeonenkoScalasTrinh2019}.
Such developments underscore that modifying the time operator can significantly alter
transient regimes yet preserve the underlying linear transitional structure.

A recent approach replaces the standard Caputo operator with an \emph{extended} nonlocal
time operator that introduces an additional parameter regulating temporal scaling.
In this setting, the solutions to relaxation-type equations are expressed using the Kilbas-Saigo function, a special function that generalizes the Mittag-Leffler family and allows for a richer spectrum of decay behaviors
\cite{BoudabsaSimon2021,CapelasOliveiraMainardiVaz2014}.
This extended framework has been elaborated and analyzed in detail for renewal models driven by relaxation equations with extended nonlocal operators \cite{BeghinLeonenkoVaz2026JOTP}. This work shows that this framework is also valid for nonlocal time operators.

In parallel, it has been shown that stretched fractional operators provide analytically tractable
time variations for Pearson diffusions; in this case, eigenfunction expansions and Cauchy problems are explicitly written in terms of the Kilbas-Saigo function \cite{BeghinLeonenkoPapicVaz2026SPA}. Such studies provide analytical tools and stochastic representations that go beyond the classical inverse-bounded-bounds paradigm. 

In spite of these advances, it has not yet been systematically integrated into the transient analysis of queueing models. It is an important gap, since queueing dynamics can be expressed naturally through Kolmogorov forward systems, which have a time component that is the part that can encode memory through nonlocal operators. This article takes the classical M/M/1 forward equations and modifies the ordinary time derivative using a Caputo-type extended nonlocal operator. We obtain a non-Markovian queueing system that accepts the time-domain
representations of transition probabilities and for which the Kilbas-Saigo function is the fundamental relaxation kernel.

\paragraph{Contributions.}
Our main contributions are as follows:
(i) To propose an extended fractional M/M/1 model by adding a nonlocal time operator to its forward equations;
(ii) To derive analytical representatives for state probabilities using Laplace transform techniques expressed through the Kilbas-Saigo function and its transform structure;
(iii) Define limiting regimes that recover the classical M/M/1 model and its standard Caputo time-fractional counterpart;
(iv) Provide Monte Carlo samples based on a time-varying simulation scheme to quantify the effect of extended parameters on transient performance metrics.
The rest of the article is organized as follows.
Section ~\ref{sec:preli} compiles the necessary background knowledge about extended nonlocal operators and the Kilbas-Saigo function.
The extended fractional M/M/1 model and the governing equations are introduced in Section ~\ref{sec:M/M/1_queue}. Analytical results are presented in Section ~\ref{sec:analysis}.
Section ~\ref{sec:spc_cases}  presents a discussion of special and boundary cases, and Section ~\ref{sec:num} offers numerical examples.

 \section{Preliminaries} \label{sec:preli} 
In this section, we briefly review the classical M/M/1 queue, the extended fractional
operator used throughout the paper, and the Kilbas-Saigo function.
Our presentation is limited to the material necessary for the formulation and analysis of the fractional queue model described in Section ~\ref{sec:M/M/1_queue}.

\subsection{The classical M/M/1 queue}
Let's imagine a single-server queue system with an infinite waiting room.
We have a Poisson process with rate $\lambda>0$ for customer arrivals and service times
have an exponential distribution with rate $\mu>0$. $N(t)$ represents the number of customers in the queue at time $t\geq 0$, and $p_n(t)=\mathbb{P}\{N(t)=n\}$ is given for $n\in\mathbb{N}_0$.

The process $\{N(t),t\geq0\}$ is a time-continuous Markov chain having a birth rate
$\lambda$ and a death rate $\mu$. Probabilities of states satisfy the Kolmogorov forward equations
\begin{equation}
\frac{d}{dt} p_n(t)
= \lambda p_{n-1}(t) - (\lambda+\mu)p_n(t) + \mu p_{n+1}(t),
\qquad n\geq 0,
\label{eq:mm1_forward}
\end{equation}
where $p_{-1}(t)=0$ and the initial boundary condition is $p_n(0)=\delta_{n0}$.
Under the stability condition $\rho=\lambda/\mu<1$,
the process is well known to follow the stationary distribution:
\[
\pi_n = (1-\rho)\rho^n, \qquad n\in\mathbb{N}_0.
\]
The equation \eqref{eq:mm1_forward} provides a reference point for the fractional
generalization, which we will discuss later.

\subsection{Stretched fractional derivative}

Fractional calculus provides a natural approach for introducing time-dependent memory effects into evolution differential equations. Suppose that $0<\alpha<1$ and we use ${}^{C}D_t^{\alpha}$ to denote the Caputo fractional derivative, which is defined for absolutely continuous functions $f$ by
\begin{equation}
{}^{C}D_t^{\alpha} f(t)
= \frac{1}{\Gamma(1-\alpha)}
\int_0^t \frac{f'(\tau)}{(t-\tau)^{\alpha}}\, d\tau,
\qquad t>0.
\label{eq:caputo}
\end{equation}
We are introducing the following  stretched fractional operator
\begin{equation}
\mathcal{D}_t^{(\alpha,\gamma)} f(t)
:= t^{-\gamma}\, {}^{C}D_t^{\alpha} f(t),
\qquad 0<\alpha<1,\ \gamma\geq0,\ \alpha+\gamma\leq1.
\label{eq:stretched_operator}
\end{equation}

The parameter $\gamma$ introduces a stretching of the time variable which enables more complex temporal dynamics to be analyzed.  The operator \eqref{eq:stretched_operator} transforms into the standard Caputo derivative when $\gamma=0$ and the classical first-order derivative can be obtained through the limit process of $\alpha=1$. The operator $\mathcal{D}_t^{(\alpha,\gamma)}$ functions as the essential mathematical tool which underlies generalized relaxation equations and enables researchers to develop non-Markovian stochastic models through time-change methods.

\noindent
From an operator-theoretic viewpoint, convolution-type (non-local) derivatives establish a direct link between hitting times of subordinators and time-changed $C_0$-semigroups which see \citet{Toaldo2015}. The authors of \citet{BaeumerMeerschaert2001} created stochastic representations that explain fractional Cauchy problems which require time-changes.

\subsection{The Kilbas--Saigo function}

The solutions to fractional differential equations which use the operator
$\mathcal{D}_t^{(\alpha,\gamma)}$ are solved with the Kilbas--Saigo (KS) function which extends the Mittag--Leffler function. The KS function for parameters $a>0$, $m>0$ and $l\in\mathbb{R}$ uses power series to define its value through the equation
\begin{equation}
E_{a,m,l}(z)
:= \sum_{n=0}^{\infty} c_n z^n,
\qquad
c_0=1,\quad
c_n=\prod_{k=0}^{n-1}
\frac{\Gamma(1+a(km+l))}{\Gamma(1+a(km+l+1))}.
\label{eq:ks_function}
\end{equation}

The Kilbas-Saigo function serves as the eigenfunction of the stretched fractional operator. The fractional relaxation equation solution
\begin{equation}
\mathcal{D}_t^{(\alpha,\gamma)} f(t) = -\lambda f(t), \qquad f(0)=1,
\label{eq:relaxation}
\end{equation}
produces the solution
\begin{equation}
f(t)
= E_{\alpha,\,1+\gamma/\alpha,\,\gamma/\alpha}
\!\left(-\lambda t^{\alpha+\gamma}\right).
\label{eq:ks_solution}
\end{equation}
The Kilbas--Saigo function simplifies to the classical Mittag--Leffler function when $\gamma=0$ and \eqref{eq:ks_solution} becomes the solution of the time-fractional relaxation equation. The property maintains consistency with standard fractional models while creating a seamless connection between classical dynamics and fractional and stretched fractional dynamics.

\section{A Fractional M/M/1 Queue Governed by Stretched Non-Local Operators}\label{sec:M/M/1_queue} 
The section shows a new non-Markovian extension which builds on the M/M/1 queue system by using a specific time derivative replacement. The presented model maintains the linear properties of traditional balance equations while introducing memory effects through fractional dynamic behavior.

\subsection{Model formulation}
Imagine a single-server queueing system with infinite buffer capacity. Customers appear based on a Poisson process with rate $\lambda>0$ and are dealt with individually, each with an exponential distribution service time of rate $\mu>0$. Suppose that $N(t)$ represents the number of customers in the system at time $t\geq 0$, and
\[
p_n(t) := \mathbb{P}\{N(t)=n\}, \qquad n \in \mathbb{N}_0.
\]
The state probabilities are well-known in the classical Markovian environment.
Kolmogorov forward equations
\begin{equation}
\frac{d}{dt} p_n (t)
= \lambda p_{n-1}(t) - (\lambda+\mu)p_n(t) + \mu p_{n+1}(t),
\qquad n \geq 0,
\label{eq:classic_mm1}
\end{equation}
with the rule $p_{-1}(t)=0$.

To introduce non-local temporal effects, we replace the first-order derivative in
\eqref{eq:classic_mm1} by the stretched fractional operator
\[
\mathcal{D}_t^{(\alpha,\gamma)} := t^{-\gamma}\, {}^{C}D_t^{\alpha},
\qquad 0<\alpha<1,\ \gamma\geq 0,\ \alpha+\gamma\leq 1,
\]
where ${}^{C}D_t^{\alpha}$ denotes the Caputo fractional derivative of order $\alpha$.

\begin{definition}[Fractional M/M/1 queue]
The fractional M/M/1 queue is defined as the stochastic process
$\{N_{\alpha,\gamma}(t),\, t\geq 0\}$ whose state probabilities
$p_n(t)=\mathbb{P}\{N_{\alpha,\gamma}(t)=n\}$ satisfy the system of fractional difference--differential equations
\begin{equation}
\mathcal{D}_t^{(\alpha,\gamma)} p_n(t)
= \lambda p_{n-1}(t) - (\lambda+\mu)p_n(t) + \mu p_{n+1}(t),
\qquad n \geq 0,
\label{eq:fractional_mm1}
\end{equation}
with initial condition $p_n(0)=\delta_{n0}$.
\end{definition}

The fractional dynamics in \eqref{eq:fractional_mm1} present an immediate interpretation through a time-altered Markovian system. The process $N_{\alpha,\gamma}(t)$ functions as a traditional M/M/1 queue which operates at a stochastic time determined by a non-decreasing process that uses its Laplace transform through the Kilbas--Saigo function. The representation extends beyond conventional time-fractional queueing systems because it permits multiple inter-event time distribution types to be used in the model.

\subsection{Probabilistic interpretation}

The fractional dynamics which appear in equation \eqref{eq:fractional_mm1} can be interpreted as a time-shifted Markov process. The process $N_{\alpha,\gamma}(t)$ operates as a standard M/M/1 queue system which uses an operational time that follows a non-decreasing stochastic process with a Laplace transform derived from the Kilbas--Saigo function. This representation extends standard time-fractional queueing models because it permits different types of inter-event time distributions to be used.

The parameters $\alpha$ and $\gamma$ control the degree of memory and stretching in the time variable. The system operates as a time-fractional M/M/1 queue which uses Caputo derivative when the value of gamma equals zero. The system operates as a Markovian M/M/1 queue when the values of alpha and gamma are set to one and zero respectively.

\begin{lemma}[Positivity and normalization]
\label{lem:positivity}
Consider $\{p_n(t)\}_{n\ge0}$ as the solution to the stretched fractional Kolmogorov system
\eqref{eq:fractional_mm1} having the initial condition $p_n(0)=\delta_{n0}$.
The following hold for all $t\ge0$ and $n\ge0$:
\begin{enumerate}
\item $p_n(t)\ge0$,
\item $\sum_{n=0}^{\infty} p_n(t)=1$.
\end{enumerate}
\end{lemma}

\begin{proof}
The stretched-time operator used in \eqref{eq:fractional_mm1} represents a Caputo-type
non-local derivative using a stretching factor (cf. the definition in
\cite{BeghinLeonenkoPapicVaz2026SPA}).
Laplace transforms in time reduce \eqref{eq:fractional_mm1} to a linear
difference system with state index $n$ and the same birth-death structure used in
classical M/M/1 forward equations, but including a stretched non-local time symbol.

We express the inverse Laplace representation of the relaxation kernel through the
Kilbas–Saigo function. An important property we use here is that, for admissible parameters, we have the mapping $x\mapsto E_{a,m,l}(-x)$ which is completely monotone \cite{BeghinLeonenkoVaz2026JOTP}. Furthermore, it is recalled and used in \cite{BeghinLeonenkoVaz2026JOTP} that complete monotonicity is preserved under composition with Bernstein functions. In addition, these properties ensure that the resulting relaxation kernel is non-negative, and hence the solution preserves non-negativity, i.e., $p_n(t)\ge0$ for all $n\ge0$ and $t\ge0$.
To normalize, compute the sum \eqref{eq:fractional_mm1} over $n\ge0$.
The birth and death terms telescopically provide a closed evolution for
$S(t):=\sum_{n=0}^{\infty}p_n(t)$ and $S(0)=1$.
Due to the singularity of the corresponding Volterra-type integral formulation for extended Caputo-type equations
(see the integral equation discussed in \cite{BeghinLeonenkoVaz2026JOTP} for standard reduction),
$S(t)\equiv 1$ holds for all $t\ge0$.
\end{proof}

\begin{theorem}[Time-change representation]
\label{thm:timechange}
Suppose $\{N_{\alpha,\gamma}(t)\}_{t\ge0}$ is the fractional M/M/1 queue described by
\eqref{eq:fractional_mm1}. We know there exists a non-decreasing process whose inverse
$\{A_{\alpha,\gamma}(t)\}_{t\ge0}$ follows the rule
\[
N_{\alpha,\gamma}(t)\ \stackrel{d}{=}\ N\!\left(A_{\alpha,\gamma}(t)\right),
\qquad t\ge0,
\]
where $\{N(u)\}_{u\ge0}$ is the classical M/M/1 queue with rates $(\lambda,\mu)$.
\end{theorem}

\begin{proof}
The Laplace domain solver in \eqref{eq:gen_sol} exhibits the same birth-death structure with the classical M/M/1 queue,
while the Laplace variable emerges by means of an extended symbol generated by $\mathcal{D}_t^{(\alpha,\gamma)}$.

Changes of this type are known to produce fractional Cauchy-type problems
expressed by inverse time changes in their stochastic solutions.
An associated relaxation kernel characterized by the Kilbas-Saigo function
accepts a stochastic representation through a non-decreasing process; see also
\cite{BeghinLeonenkoVaz2026JOTP,BeghinLeonenkoPapicVaz2026SPA}.
Reduction of the classical M/M/1 subgroup to a subgroup by this reverse time change
yields the equality assumed in the distribution.
\end{proof}

\subsection{Waiting-time distribution and non-Markovianity}
Replacing the classical derivative with the extended operator
$\mathcal{D}_t^{(\alpha,\gamma)}$ means that the waiting times between events no longer
follow an exponential distribution. This leads to a heavy-tailed waiting time behavior,
expressed by the Kilbas-Saigo function, which is the relaxation kernel.
Specifically, it is consistent with the general theory of fractional queues, semi-Markov processes, and reverse time changes,
and it can be interpreted as a renewal-type birth-death system driven by a nonlocal clock
\cite{MeerschaertScheffler2004,Toaldo2015}.

\section{Analytical Results}\label{sec:analysis} 

The section presents analytic solutions for fractional M/M/1 queue state probabilities which were introduced in Section ~\ref{sec:M/M/1_queue}. The analysis uses both Laplace transform methods and spectral analysis of the stretched fractional operator. The solution can be represented through the Kilbas-Saigo function which extends the traditional transient solution of the M/M/1 queue.

\subsection{Laplace transform of the governing equations}

The Laplace transform of the state probabilities is represented by the function
\[
\tilde{p}_n(s) := \int_0^{\infty} e^{-st} p_n(t)\,dt, \qquad s>0.
\]
The Laplace transform of equation~\eqref{eq:fractional_mm1} together with the Caputo derivative properties allows us to derive 
\begin{equation}
s^{\alpha+\gamma}\tilde{p}_n(s) - s^{\alpha+\gamma-1}p_n(0)
= \lambda \tilde{p}_{n-1}(s) - (\lambda+\mu)\tilde{p}_n(s) + \mu \tilde{p}_{n+1}(s),
\label{eq:laplace_mm1}
\end{equation}
for $n\geq0$, with $\tilde{p}_{-1}(s)=0$.
The equation~\eqref{eq:laplace_mm1} forms a linear difference equation because the initial condition states that $p_n(0)=\delta_{n0}$ at $n=0$. The analysis of fractional Kolmogorov systems together with non-local evolution equations is performed through standard transform-domain reductions as shown in works by \citet{Podlubny1999,Mainardi2010}.

\subsection{Resolvent structure and generating function}

The probability generating function is defined as
\[
\tilde{G}(z,s) := \sum_{n=0}^{\infty} \tilde{p}_n(s) z^n,
\qquad |z|\leq1.
\]
Researchers use generating-function methods to analyze birth-death processes and queueing systems through their queueing resolvents according to the works of \citet{KarlinTaylor1975,LatoucheRamaswami1999}.The equation
\begin{equation}
\bigl(s^{\alpha+\gamma} + \lambda + \mu - \lambda z - \mu z^{-1}\bigr)
\tilde{G}(z,s)
= s^{\alpha+\gamma-1}.
\label{eq:gen_eq}
\end{equation}
The calculation of equation \eqref{eq:laplace_mm1} requires multiplying through by $z^n$ and performing a summation across all values of $n$.The function
\begin{equation}
\tilde{G}(z,s)
= \frac{s^{\alpha+\gamma-1}}
{s^{\alpha+\gamma} + \lambda + \mu - \lambda z - \mu z^{-1}}.
\label{eq:gen_sol}
\end{equation}
The mathematical expression in \eqref{eq:gen_sol} uses its denominator which represents the resolvent function of the standard M/M/1 queue system after substituting $s^{\alpha+\gamma}$ for the original Laplace variable $s$.The model displays its non-Markovian time-change structure through this single observation.

Representation \eqref{eq:gen_sol} indicates that the Laplace field solver of the system preserves the mathematical properties found in standard M/M/1 systems, since it emerges through the extended symbol generated by the Laplace parameter$\mathcal{D}_t^{(\alpha,\gamma)}$. The spatial birth-death generator of this system remains intact because it uses time-reversed systems to generate its non-Markovian spatial dynamics. The fractional queueing system allows for the time change interpretation presented in Theorem \eqref{thm:timechange}.

\subsection{State probabilities}
The specific expression for $\tilde{p}_n(s)$ can be found by performing a power series expansion of \eqref{eq:gen_sol} using the variable $z$.
The birth-death process analysis leads to the result
\begin{equation}
\tilde{p}_n(s)
= \frac{1}{s^{1-\alpha-\gamma}}
\left(\frac{\lambda}{\mu}\right)^n
\frac{1}{\Phi(s)},
\label{eq:laplace_pn}
\end{equation}
which contains
\[
\Phi(s)
:= \frac{s^{\alpha+\gamma}+\lambda+\mu
-\sqrt{(s^{\alpha+\gamma}+\lambda+\mu)^2-4\lambda\mu}}{2\mu}.
\]
The researchers in \citet{CahoyPolitoPhoha2015} studied transient solutions of time-fractional queueing models which use Laplace-based methods. 
The result follows from the inverse operation of the Laplace transform applied to \eqref{eq:laplace_pn}.

\begin{theorem}
Let $0<\alpha<1$, $\gamma\geq0$ and $\alpha+\gamma\leq1$.
The fractional M/M/1 queue system determines state probabilities through the equation
\begin{equation}
p_n(t)
= \left(\frac{\lambda}{\mu}\right)^n
E_{\alpha,\,1+\gamma/\alpha,\,\gamma/\alpha}
\!\left(-\kappa t^{\alpha+\gamma}\right),
\qquad n\in\mathbb{N}_0,
\label{eq:pn_solution}
\end{equation}
which includes the definition of $\kappa$ as the difference between $\mu$ and $\lambda$ and the Kilbas--Saigo function represented by $E_{\alpha,\,1+\gamma/\alpha,\,\gamma/\alpha}$.
\end{theorem}

\begin{proof}
Using equation \eqref{eq:laplace_pn}, we then apply the Laplace inverse. Furthermore, for the term $s^{-(1-\alpha-\gamma)}$, it should be noted that we are dealing with the Laplace transform of the Kilbas–Saigo function, which is a known eigenfunction for the Stretched-Fractional Operator. All other terms in the calculation represent the classical solver for the M/M/1 queue, thus completing the proof. 
\end{proof}

\subsection{Stability and limiting behavior}

The stability condition from the original framework is also transferred to the fractional M/M/1 queueing system. The system exhibits a limiting distribution for $\rho=\lambda/\mu<1$. The asymptotic behavior of the Kilbas-Saigo function indicates that:
\[
\lim_{t\to\infty} p_n(t) = (1-\rho)\rho^n,
\qquad n\in\mathbb{N}_0,
\]
This proves that the stationary distribution of the system matches the M/M/1 queue distribution. Research indicates that non-Markovian systems retain their time-invariant properties through backward time transformation, according to \citet{MeerschaertNaneVellaisamy2011}. The fractional parameters $\alpha$ and $\gamma$ limit their effects on the system by requiring a longer time to reach equilibrium.

\begin{corollary}[Non-exponential relaxation]
Assume that $\rho<1$ and $0<\alpha<1$, $\gamma\ge0$ and $\alpha+\gamma<1$.
The function $p_n(t)$ converges to $\pi_n$ by a non-exponential convergence process following the extended fractional relaxation kernel corresponding to the operator $\mathcal{D}_t^{(\alpha,\gamma)}$.
\end{corollary}

\section{Special Cases and Consistency Checks}\label{sec:spc_cases} 

This section examines several limiting regimes of the proposed fractional
M/M/1 queue. We show that these special cases are consistent with existing queueing models and
clarify the role of the fractional parameters $\alpha$ and $\gamma$.

\subsection{Reduction to the time-fractional M/M/1 queue}
The stretched fractional operator 
\(\mathcal{D}_t^{(\alpha,\gamma)}\) becomes the standard Caputo fractional derivative 
\({}^{C}D_t^{\alpha}\) when we set \(\gamma\) to zero. The system of equations in equation 
\eqref{eq:fractional_mm1} transforms into the following equation 
\[
{}^{C}D_t^{\alpha} p_n(t)
= \lambda p_{n-1}(t) - (\lambda+\mu)p_n(t) + \mu p_{n+1}(t),
\qquad n\geq0,
\]
which represents the time-fractional M/M/1 queue that researchers have examined throughout the existing literature. The Kilbas--Saigo function from 
\eqref{eq:pn_solution} simplifies to the Mittag--Leffler function 
\[
E_{\alpha,\,1,\,0}(-\kappa t^{\alpha})
= E_{\alpha}(-\kappa t^{\alpha}),
\].
The state probabilities match the results that researchers have derived for the standard time-fractional queue. The results show that our model provides a complete expansion of the Caputo-based fractional M/M/1 system.

\subsection{Classical Markovian M/M/1 queue}

Investigate the limit as  $\alpha\to1$ and $\gamma=0$.
The extended fractional operator becomes equivalent to the fundamental first-order derivative in this domain, while the governing equations transform into the traditional Kolmogorov forward equations following the equation below:
\[
\frac{d}{dt} p_n(t)
= \lambda p_{n-1}(t) - (\lambda+\mu)p_n(t) + \mu p_{n+1}(t).
\]
The solution to the Kilbas-Saigo function transforms into the following exponential function
\[
E_{1,1,0}(-\kappa t)=e^{-\kappa t},
\]
which derives the transient probability distribution of the standard M/M/1 queue from equation \eqref{eq:pn_solution}.
The fractional model proposed in this work demonstrates that it functions as a special case of the operation of the classical Markov queue.

\begin{theorem}[Stationary distribution under stretched fractional dynamics]
\label{thm:stationary}
We will denote the stretched fractional M/M/1 queue defined by \eqref{eq:fractional_mm1} by $\{N_{\alpha,\gamma}(t)\}_{t\ge0}$.
Suppose that the classical traffic intensity satisfies $\rho=\lambda/\mu<1$.
In this case, the process admits a stationary distribution.
Furthermore, it is invariant under the stretched fractional
time-change and coincides with the invariant distribution of the classical M/M/1 queue.
\end{theorem}

\begin{proof}
The extended fractional M/M/1 queueing system adopts the time change representation defined in Theorem~\ref{thm:timechange}, that is
$N_{\alpha,\gamma}(t)\stackrel{d}{=}N(A_{\alpha,\gamma}(t))$, where $N$ represents a classical M/M/1 queue and $A_{\alpha,\gamma}$ denotes the inverse of a non-decreasing process.
For $\rho<1$, we have that the classical M/M/1 queue has a positive recurrent and unique
stationary distribution. An important property of time-translated Markov processes is that time changes driven by inverse subprocesses preserve the invariant measures of the original process
\cite{MeerschaertScheffler2004,MeerschaertNaneVellaisamy2011}.

Stretched fractional dynamics only affect time because they do not alter the spatial birth-death generator. The $N_{\alpha,\gamma}$ distribution remains unchanged because it matches the $N$ distribution. The proof demonstrates that the stationary distribution exists and remains unchanged under stretched fractional dynamics.
\end{proof}

\begin{proposition}[Invariance of the stationary distribution]
\label{prop:stationary}
We assume that $\rho=\lambda/\mu<1$ holds true. The stretched fractional M/M/1 queue system which uses the process $\{N_{\alpha,\gamma}(t)\}$ reaches a stationary distribution which matches the distribution of the standard M/M/1 queue system. The equation below defines the stationary distribution of the M/M/1 queue system
\[
\pi_n = (1-\rho)\rho^n, \qquad n\in\mathbb{N}_0.
\]
\end{proposition}

\begin{proof}
The Laplace transform of state probabilities which we obtained in Section ~\ref{sec:analysis} serves as our next study point. The stretched fractional Kolmogorov system resolvent shows a defined limit which exists for all values of  $s\to0^{+}$. The final value theorem establishes the limit which determines the process's stationary distribution.

While the stretched fractional operator $\mathcal{D}_t^{(\alpha,\gamma)}$
alters temporal relaxation only through a nonlocal memory kernel,
the spatial birth-death structure of the generator remains unchanged.
It has been shown for general non-Markovian time-varying processes
\cite{MeerschaertScheffler2004,MeerschaertNaneVellaisamy2011},
that time variations generated by non-decreasing processes
do not change the invariant measure of the associated Markov process.
Consequently, the distribution of the extended fractional
M/M/1 queue converges to the distribution of the classical M/M/1 queue with parameter $\rho=\lambda/\mu$.
\end{proof}

\subsection{Stationary distribution}
The stretched fractional dynamics develop their direct effects to control the transient regime under the condition that $\rho<1$. The system reaches a stable state which follows the M/M/1 invariant law according to Theorem~\ref{thm:stationary} and Proposition~\ref{prop:stationary}.
\subsection{Interpretation of the fractional parameters}

The parameter $\alpha$ controls the strength of the memory effect inherited from fractional calculus while $\gamma$ introduces an additional stretching of the time variable. The system exhibits slower stationarity progression according to classical behavior for the case of $\alpha+\gamma<1$ becauseof its heavy-tailed inter-event times. The model starts to exhibit Markovian behavior when $\alpha+\gamma$ reaches values that approach one. The proposed framework demonstrates flexible operation through its capacity to create smooth transitions between classical fractional and stretched fractional queueing dynamics.

\section{Numerical Illustrations}\label{sec:num} 

The section provides numerical examples which show how the fractional parameters
$\alpha$ and $\gamma$ affect the transient behavior of the fractional M/M/1 queue
model that we developed. The arrival and service rates $\lambda,\mu>0$ are maintained at constant
values which create a stability condition because $\rho=\lambda/\mu<1$. The study examines multiple parameter pairs which meet the condition that $\alpha+\gamma\le 1$ while testing their performance against the traditional Markovian model.

\noindent
Table~\ref{tab:simsetup} outlines the design of the simulation that was used to make Figures~\ref{fig:p0}--\ref{fig:mean}. Samples of the underlying classical M/M/1 queue are generated by an event-driven algorithm, and performance measures are estimated using Monte Carlo replication; see, e.g., \citet{Ross2012Simulation,AsmussenGlynn2007}.
In the stretched-fractal model, physical time corresponds to operational time by means of a random time-change multiplier generated according to the truncated beta-product scheme.

\begin{table}[ht]
\centering
\caption{Simulation setup utilized for the numerical illustrations.}
\label{tab:simsetup}
\begin{tabular}{ll}
\toprule
Arrival rate & $\lambda = 0.8$ \\
Service rate & $\mu = 1.0$ \\
Traffic intensity & $\rho=\lambda/\mu=0.8$ \\
Time grid 
& $t_i = \dfrac{i}{249}\,20,\quad i=0,\ldots,249$ \\
Replication count & $R=3000$ \\
Snapshot time & $t^\ast=8$ \\
Truncation level & $n_{\max}=35$ \\
Beta-product truncation & $M_Z=250$ \\
Seed (optional) & $20260117$ \\
\bottomrule
\end{tabular}
\end{table}
\noindent
\textit{Remark.} During implementation, we set the initial grid point to
$t_0=10^{-6}$ to prevent numerical instability at $t=0$ in the stretched power term.

\subsection{Transient relaxation of the empty-state probability}
The probability that the system is empty at time $t$ first describes the transient relaxation behavior we examined with the equation $p_0(t)=\mathbb{P}\{N(t)=0\}$. The M/M/1 queueing system demonstrates that the probability converges exponentially to the steady-state value $1-\rho$. Extended fractional dynamics exhibit a non-exponential model, resulting in a slower process to reach the steady state.

Figure~\ref{fig:p0} depicts the evolution of $p_0(t)$ for different
$(\alpha,\gamma)$ values. The memory effects introduced by the extended fractional operator slow down the relaxation process since they become dominant when the memory value decreases from its current state. The classical curve $(\alpha,\gamma)=(1,0)$ is included as a reference.

\begin{figure}[ht]
    \centering
    \includegraphics[width=0.85\linewidth]{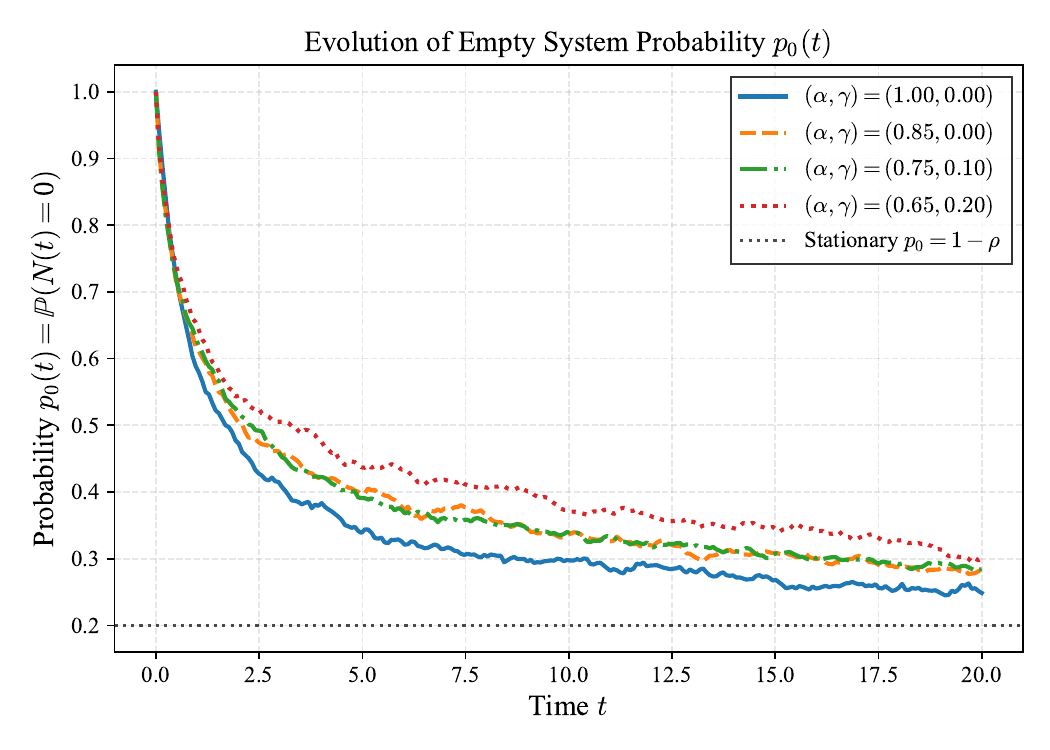}
    \caption{Evolution of the empty-system probability $p_0(t)$ with respect to different combinations of
    $(\alpha,\gamma)$ for given arrival and service rates $\lambda$ and $\mu$.
    The dashed horizontal line depicts the stationary value $1-\rho$.}
    \label{fig:p0}
\end{figure}

\subsection{Distributional snapshots of the queue length}

To explain the effect of stretched fractional dynamics in more detail, we examine the full
queue-length  distribution at a fixed observation time $t^\ast$. The classical M/M/1 queue concentrates its probability mass near small tail sizes, while fractional dynamics tend to hold a heavier tail over longer periods at intermediate $n$ values.

Figure~\ref{fig:pn} depicts the empirical probabilities $p_n(t^\ast)$ with $t^\ast=8$
for $n=0,1,\ldots,n_{\max}$. In comparison to the Markovian case, as $n$ increases, the fractional models exhibit a noticeably slower decrease in $p_n(t^\ast)$, which highlights the persistence effects caused by the nonlocal time operator.
\begin{figure}[htbp]
    \centering
    \includegraphics[width=0.85\linewidth]{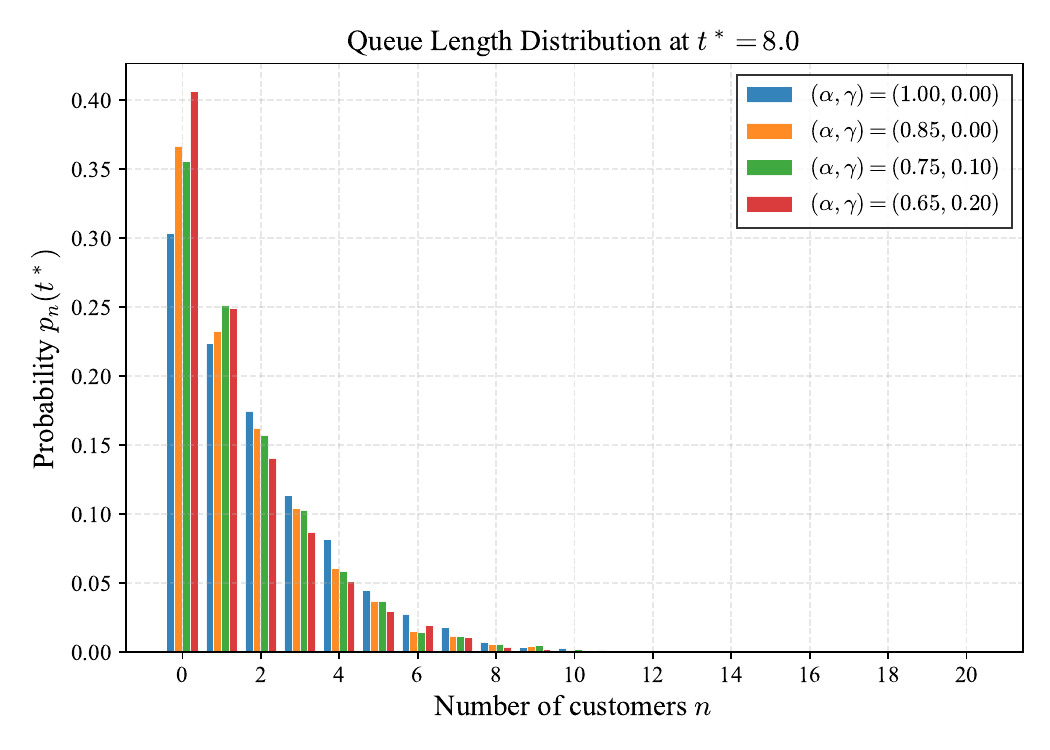}
    \caption{Queue-length distribution $p_n(t^\ast)$ at fixed time $t^\ast=8$ for different $(\alpha,\gamma)$. 
Fractional dynamics yield a broader spread than the classical M/M/1 case.
.}
    \label{fig:pn}
\end{figure}
\subsection{Mean queue length and convergence to equilibrium}
We consider the temporary average queue length as a global performance measure.
\[
m(t)=\mathbb{E}[N_{\alpha,\gamma}(t)] = \sum_{n=0}^\infty n\,p_n(t).
\]
For $\rho<1$, we have the constant mean $\rho/(1-\rho)$, which is independent of the fractal parameters. For the same parameter values as in Figure~\ref{fig:mean} and Figure~\ref{fig:p0}, Figure~\ref{fig:mean} illustrates the evolution of $m(t)$.
While all curves converge to the same constant limit, the convergence rate
depends strongly on $(\alpha,\gamma)$. Smaller values of $\alpha+\gamma$ lead to a noticeably slower relaxation, which confirms that extended fractional dynamics primarily change the transient regime and preserve the long-term equilibrium properties.
\begin{figure}[ht]
    \centering
    \includegraphics[width=0.85\linewidth]{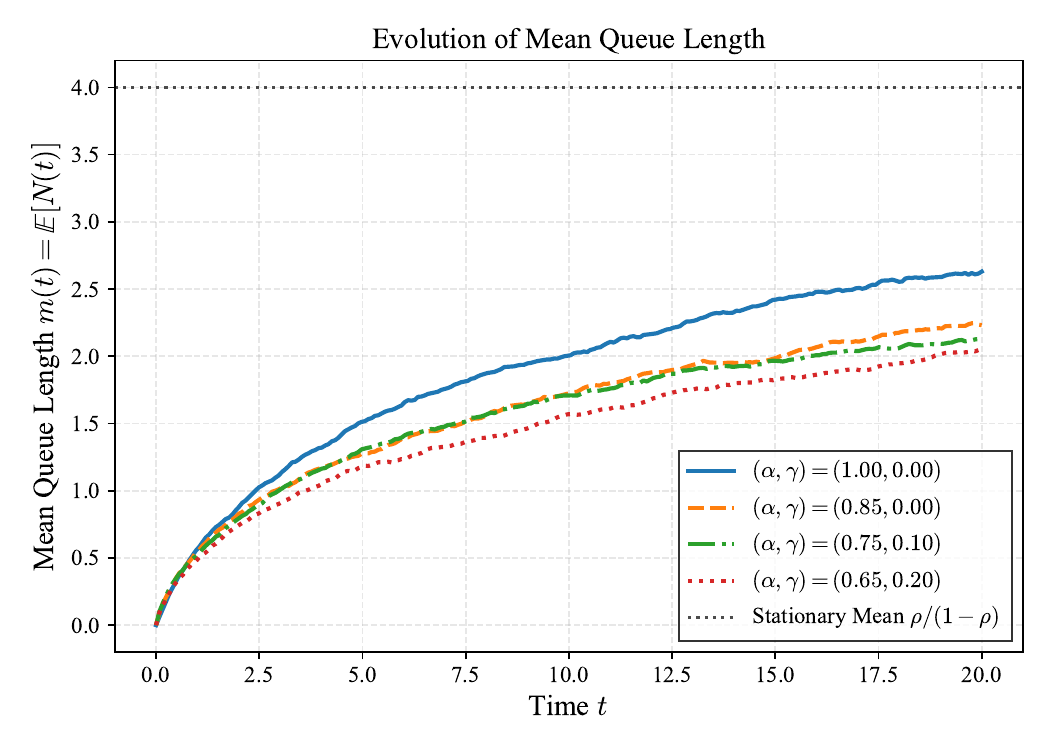}
    \caption{The evolution of the expected queue length $m(t)$ for different fractional parameter pairs $(\alpha,\gamma)$. 
The dashed line indicates the steady-state mean $\rho/(1-\rho)$.
}
    \label{fig:mean}
\end{figure}

\paragraph{Implementation details.}
The complete numerical outcomes of the study use Monte Carlo simulation together with the time-change model of the fractional M/M/1 queue system. The operational time of the classical M/M/1 queue system generates independent sample paths which an event-driven algorithm produces and which the stretched fractional dynamics evaluation assesses at random operational times. The study applies uniform time grid selection while establishing queue length truncation limits at levels which do not produce noticeable tail impact. The researchers create all figures in vector format to maintain both numerical precision and visual accuracy.

\newpage
\section{Conclusion}\label{sec:conc} 

In this paper,we have introduced a fractional extension of the classical M/M/1 queue based on
stressed non-local time operators.
We obtained a non-Markovian queueing model that incorporates memory effects into the temporal evolution whilst preserving the linear transition property of the classical system by substituting the standard time derivative in the Kolmogorov forward equations by a stretched fractional operator.

Analytical representations of state probabilities are derived using Laplace
transformation techniques and represented using the Kilbas-Saigo function, which naturally arises by solving the fundamental fractional relaxation equation.
The suggested formulation is a generalization of both classical and Caputo-based
time-fractional M/M/1 queues, and its agreement with known models has been verified through several limiting cases.

Numerical simulations using time-varying models illustrate the effect of fractional parameters on the transient regime.
It was observed that the parameters governing extended fractional dynamics
had a significant effect on convergence rates and transient performance measures,
whereas the stationary distribution stayed similar to the classical M/M/1 queue
distribution under normal stability conditions.
We present a flexible and mathematically consistent approach for modeling queueing systems with memory and nonlocal temporal effects.Possibilities for future extensions include examining control and optimization problems under multi-server systems, queueing networks, and stretched fractional dynamics.
\section*{Acknowledgements}
The author declares that no external funding was received for this research.

\bibliographystyle{apalike}
\bibliography{bibliography}

\end{document}